\documentclass[11pt,fleqn]{article}

\usepackage{fullpage}
\usepackage{amsthm, amssymb, amsmath}
\usepackage{dsfont}
\usepackage{color}
\usepackage{xspace}

\usepackage{ifpdf}
\ifpdf    %
\usepackage{hyperref}
\else    %
\usepackage[hypertex]{hyperref}
\fi

\newtheorem{theorem}{Theorem}[section]
\newtheorem{lemma}[theorem]{Lemma}

\newtheorem{claim}[theorem]{Claim}

\newtheorem{corollary}[theorem]{Corollary}

\newcommand{\zo}{\{0,1\}}

\newcommand{\R}{{\mathbb{R}}}

\newcommand{\veps}{\varepsilon}
\newcommand{\eps}{\veps}
\renewcommand{\epsilon}{\varepsilon}
\newcommand{\tO}{\tilde{O}}

\newcommand{\EX}{{\mathbb E}}

\DeclareMathOperator{\Var}{Var}

\DeclareMathOperator{\sk}{{\bf sk}}
\DeclareMathOperator{\est}{{\bf est}}

\DeclareMathOperator{\wt}{wt}
\DeclareMathOperator{\IC}{IC}
\newcommand{\calP}{{\cal P}}
\newcommand{\Lhigh}{L_{\mathrm{high}}}
\newcommand{\indic}{\mathds{1}}

\providecommand{\eqdef}{:=}
\providecommand{\minn}[1]{\min\{{#1}\}}

\providecommand{\ceil}[1]{\lceil #1 \rceil}
\providecommand{\floor}[1]{\lfloor #1 \rfloor}
\providecommand{\abs}[1]{\lvert#1\rvert}
\providecommand{\card}[1]{\lvert#1\rvert}

\providecommand{\aset}[1]{\{ #1 \}}
\newcommand{\whp}{with high probability\xspace}
\newcommand{\half}{\tfrac12}

\def\compactify{\itemsep=0pt \topsep=0pt \partopsep=0pt \parsep=0pt}

\title{The Sketching Complexity of Graph Cuts%
\footnote{The current version differs slightly from an earlier one
(arXiv:1403.7058v1).
First, the lower bound for the number of edges in $(1+\veps)$-cut sparsifiers
is improved, and now our bound is tight.
Second, we retract our earlier claim that the sparsification algorithm
can be performed in two passes of streaming over the graph edges.
} }

\author{
Alexandr Andoni\thanks{%
Work done in part while the author was
at Microsoft Research Silicon Valley.
Email: \texttt{andoni@mit.edu}
}
\and
Robert Krauthgamer\thanks{%
Weizmann Institute of Science, Israel. 
Work supported in part by a US-Israel BSF grant \#2010418,
an Israel Science Foundation grant \#897/13,
and by the Citi Foundation.
Part of the work was done at Microsoft Research Silicon Valley.
Email: \texttt{robert.krauthgamer@weizmann.ac.il} 
}
\and
David P. Woodruff\thanks{%
IBM Almaden, USA.
Email: \texttt{dpwoodru@us.ibm.com}}
}

\begin{document}

\maketitle 

\begin{abstract}
We study the problem of sketching an input graph, so that given the
sketch, one can estimate the weight of any cut in the graph
within factor $1+\epsilon$. 
We present lower and upper bounds on the size of a randomized sketch, 
focusing on the dependence on the accuracy parameter $\epsilon>0$.

First, we prove that for every $\epsilon > 1/\sqrt n$,
every sketch that succeeds (with constant probability) 
in estimating the weight of all cuts $(S,\bar S)$ in an $n$-vertex graph 
(simultaneously), must be of size $\Omega(n/\epsilon^2)$ bits.
In the special case where the sketch is itself a weighted graph 
(which may or may not be a subgraph) 
and the estimator is the sum of edge weights across the cut in the sketch, i.e., a cut sparsifier, 
we show the sketch must have $\Omega(n/\epsilon^2)$ edges, which is optimal. 
Despite the long sequence of work on graph sparsification, 
no such lower bound was known on the size of a cut sparsifier. 

We then design a randomized sketch that, given $\epsilon\in(0,1)$ 
and an edge-weighted $n$-vertex graph, 
produces a sketch of size $\tilde O(n/\epsilon)$ bits, 
from which the weight of any cut $(S,\bar S)$ can be reported, 
with high probability, within factor $1+\epsilon$.
The previous upper bound is $\tilde O(n/\epsilon^2)$ bits,
which follows by storing a cut sparsifier 
as constructed by Bencz{\'u}r and Karger \cite{BK96} 
(or by followup work \cite{SS11,BSS12,FHHP11,KP12}). To obtain this
improvement, we critically use both that the sketch need only be correct
on each fixed cut with high probability (rather than on all cuts), and that
the estimation procedure of the data structure can be arbitrary (rather than
a weighted subgraph). 
We also show a lower bound of $\Omega(n/\epsilon)$ bits for
the space requirement of any data structure achieving this guarantee. 
As an application of our data structure,
we show how to compute, in a communication-efficient manner, 
a global minimum cut of a graph that is distributed across multiple servers.
\end{abstract}

\thispagestyle{empty}
\newpage
\setcounter{page}{1}

\section{Introduction}

In 1996 Bencz\'ur and Karger \cite{BK96} introduced \emph{cut
  sparsifiers}, a remarkable and very influential notion: given a
graph $G=(V,E,w)$ with $n=\card{V}$ vertices, $m=\card{E}$ edges and
edge weights $w:E\to\R_+$, together with a desired error parameter
$\eps\ge 0$, the cut sparsifier is a sparse graph $H$ on the same $n$
vertices, in fact a subgraph of $G$ but with different edge weights,
such that every cut in $G$ is $(1+\eps)$-approximated by the
corresponding cut in $H$.  Specifically, they show there always exists
such a sparsifier $H$ with only $O((n\log n)/\eps^2)$ edges,
potentially much fewer than in the original graph $G$, and moreover $H$ can
be constructed in time $O(m\log^2 n)$.  This construction has had 
a tremendous impact on cut problems in graphs, see
e.g.\ \cite{BK96,BK02, KL02, sherman2009breaking, madry2010fast}. Often, one
first replaces the original graph with a cut sparsifier and then solves an
optimization problem on the sparsifier.

The notion of cut sparsifiers turned out to be very influential.
Followup work offers alternative constructions and refined bounds 
\cite{GRV09-spanning, FHHP11, KP12}, as well as a
generalization of the notion, termed {\em spectral sparsifiers}
\cite{ST04a, ST11, SS11, BSS12}. A spectral sparsifier is a graph $H$ such
that the quadratic form associated with the Laplacian of $H$ 
provides $(1+\eps)$-approximation to that of $G$. 
The bound of $O(n/\eps^2)$ by Batson, Spielman, and Srivastava~\cite{BSS12}
improves also over the cut sparsifiers of Bencz\'ur and Karger~\cite{BK96} 
(in terms of graph size, but not construction time). 
These spectral sparsifiers have been instrumental
in obtaining the first near-linear time algorithm for solving SDD linear
systems \cite{ST04a}, and as such have led to many beautiful
ideas down the road, see also \cite{KMP10, KMP11}.

For some applications, the
value of $\eps$ may be determined so as to obtain the smallest approximation
error under the sole constraint that the sketch fits in memory. If, say, $\eps$ is
$1\%$ then $1/\eps^2 = 10000$ and a quadratic dependence 
might be prohibitive on large-scale graphs. 
It is therefore natural to ask whether the dependence on $\eps$ 
in the cut-sparsifier size can be improved. 
The Alon-Boppana theorem and
\cite{Alon97-expansion} suggest the answer is no --- if the sparsifier $H$ is
constrained to be a regular graph with uniform edge weights, 
then $\Omega(n/\eps^2)$ edges are required to approximate all the cut
values of the complete graph 
(for similar results about spectral sparsifiers, see \cite{Nil91, BSS12}).  
This lower bound is optimal \emph{quantitatively} -- for the complete graph, the sparsifier $H$ can be a regular expander graph \cite{LPS88-ramanujan, MSS13} -- 
but it is unsatisfactory \emph{qualitatively}, as it does not cover
the sparsifier $H$ being any reweighted subgraph of $G$, i.e., those that 
are not regular or of non-equal
edge weights. 
Moreover, it may be desirable to construct --- or to rule out ---
an arbitrary graph $H$ on the same vertex-set $V$
(but not necessarily a subgraph of $G$).
We address this issue by proposing and studying an even more general concept
than cut sparsifiers.

\subsection{Our Contribution}

Define a \emph{graph-cuts sketch} to be a data structure 
that can report the weight of any cut in an input graph $G$.
In the preprocessing stage, the sketch is constructed from $G=(V,E,w)$;
then, the sketch is used to answer, without access to $G$, the following queries: 
given any $S\subset V$, report the weight of the respective cut 
$w(S,\bar S) \eqdef \sum_{e\in E: \card{e\cap S}=1} w(e)$.
Throughout, we shall consider a reported answer to be correct
if it is a $(1+\eps)$-approximation of the true cut weight $w(S,\bar S)$.

A sketch clearly generalizes the notion of a cut sparsifier,
because the data structure can simply build and store a cut sparsifier $H$,
and then answer any query by reporting the weight of the corresponding cut in $H$.
By similar arguments, a sketch captures much more complicated scenarios,
such as an arbitrary graph $H$ on the same vertex-set (mentioned above)
or possibly on an augmented vertex-set $V'\supset V$,
or even storing multiple graphs $H_1,H_2,H_3$ and applying some computation 
(e.g., taking the minimum between their respective cut values).

When the sketch is randomized, 
we shall require that the sketch succeeds with high probability 
(say at least $0.9$, which can be amplified to $1-1/n$ by independent repetitions), 
which can be interpreted in two different ways:
(1) with high probability, for all cuts the answer is correct; or
(2) for each cut, with high probability the answer is correct.
In the compressed sensing literature, these guarantees are often termed
``for all'' and ``for each'', respectively.

\paragraph{Lower Bound.} 

We prove that every, possibly randomized, ``for all'' sketch
must have size $\Omega(n/\eps^2)$ bits. 
In the special case that the sketch is a weighted graph itself 
(which may or may not be a subgraph of $G$) and
the estimator is the sum of edge weights across the cut in the sketch, i.e.,
a cut sparsifier, we show the sketch must have $\Omega(n/\veps^2)$ edges, which is
optimal. 
Thus, we strengthen the lower bounds of Alon-Boppana and of~\cite{Alon97-expansion} 
in the following sense --- if every graph $G$ has 
a (not necessarily regular, and possibly weighted) cut sparsifier
with at most $s(n,\eps)$ edges, then $s=\Omega(n/\eps^2)$. 
Despite the long sequence of work on graph sparsifiers, the optimal
number of edges in a cut sparsifier was previously unknown. 

As our $\Omega(n/\eps^2)$ bit lower bound holds for every sketching algorithm, 
it covers also the potentially more sophisticated operations mentioned above.
Our lower bound applies also to sketches that can approximate the 
quadratic form of the Laplacian of $G$ within factor $1+\eps$, 
as such queries are known to include all cuts of $G$. 
Thus, in addition to our lower bound for cut sparsifiers, 
our result generalizes an $\Omega(n/\eps^2)$ edge lower bound 
of \cite{BSS12} for spectral sparsifiers of the complete graph, 
in the sense that our result holds for all sketching algorithms, 
not only those which approximate $G$ by another graph $H$. 

Several streaming algorithms produce cut or spectral sparsifiers 
with $O(n/\eps^2)$ edges using $\tO(n/\eps^2)$ bits of space in a stream
\cite{ag09, KL13-sparsification, goel2010graph, agm12, GKP12, KLMMS14}. 
Our lower bound shows these algorithms have optimal space, 
up to logarithmic factors, and the sparsifier they produce is optimal,
up to a constant factor, in its number of edges.

Our lower bounds are shown via communication complexity. 
In addition to strengthening prior work, 
it introduces to this area new tools from communication complexity,
an approach that will hopefully be useful also in related contexts. 
In fact, we do not know of a proof that does not use communication complexity. 
Our lower bound was recently used in \cite{kk15} to argue that 
their construction of $(1+\eps)$-cut-sparsifiers for $r$-uniform hypergraphs
has near-optimal size. 

\paragraph{Upper Bound.} 
Perhaps surprisingly, we 
construct a randomized ``for each'' sketch that achieves $(1+\eps)$-approximation 
using size $O(\tfrac{n}{\eps}(\log n)^{O(1)})$.
This result shows that a {\em linear} dependence on $1/\eps$
can be achieved if we do not require that all cuts be estimated simultaneously. 
We also show that our sketch size is optimal up to polylogarithmic factors,
as a ``for each'' sketch requires size $\Omega(n/\eps)$;
see Section \ref{sec:nEpsLB} for a proof.

We demonstrate how a ``for each'' sketch can be useful algorithmically
despite its relaxed guarantees compared to a cut sparsifier.
In particular, we show how to $(1+\eps)$-approximate the global minimum cut 
of a graph whose edges
are distributed across multiple servers. Distributed large-scale graph
computation has received recent attention, where protocols for 
distributed minimum spanning tree, breadth-first search, shortest paths,
and testing connectivity have been studied, among other problems
(see, e.g., \cite{knpr15,wz13}). In our case, each server locally computes
the ``for each'' data structure on its subgraph (for accuracy $\eps$), 
and sends it to a central server. 
Each server also computes a classical cut sparsifier, with 
fixed accuracy $\eps' = 0.2$, and sends it to the central server. Using that cut-sparsifiers
can be merged, the central server obtains a $(1 \pm \eps')$-approximation to all cuts
in the union of the graphs. By a result of Henzinger and Williamson \cite{hw96}
(see also Karger \cite{karger2000minimum}), there are only $O(n^2)$
cuts strictly within factor $1.5$ of the minimum cut, and they can be found efficiently
from the sparsifier (see \cite{karger2000minimum} for an $\tO(n^2)$ time way of implicitly
representing all such cuts). The central server then 
evaluates each ``for each'' data structure on each of these cuts, and sums up the
estimates to evaluate each such cut up to factor $1+\eps$, 
and eventually reports the minimum found. Note
that the ``for each'' data structures can be assumed to be correct with probability $1-1/n^3$
for any fixed cut (by independent repetitions), 
and therefore correct on all $O(n^2)$ candidate cuts. 

Recently, Chen et al. \cite{CQWZ14} showed how to generalize our ``for each'' upper bound
from cut queries to spectral queries, namely, 
to the case of preserving $x^TLx$ for a fixed $x$ up to a $1+\eps$ factor, where $L$ is
the Laplacian of a graph. Their data structure uses $\tO(n/\eps^{1.6})$ bits
of space.

\subsection{Techniques}

\paragraph{Lower Bound.} 

We prove our $\Omega(n/\eps^2)$ bit lower bound for arbitrary data structures using 
communication complexity. We then show how to obtain an $\Omega(n/\eps^2)$
edge lower bound for cut sparsifiers by encoding a sparsifier in a careful way
so that if it had $o(n/\eps^2)$ edges, it would violate an $\Omega(n/\eps^2)$ bit
lower bound in the communication problem. 

For the $\Omega(n/\eps^2)$ bit lower bound, the natural
thing to do would be to give Alice a graph $G$, and Bob a cut
$S$. Alice produces a sketch of $G$ and sends it to Bob, who must
approximate the capacity of $S$. The communication cost of this
problem lower bounds the sketch size. However, as we just saw, Alice
has an upper bound with only $\tO(n/\eps)$ bits of communication. We
thus need for Bob to solve a much harder problem which uses the fact
that Alice's sketch preserves all cuts.

We let $G$ be a disjoint union of $\eps^2 n/2$ graphs $G_i$, where
each $G_i$ is a bipartite graph with $\frac{1}{\eps^2}$ vertices in
each part. Each vertex in the left part is independently connected to
a random subset of half the vertices in the right part. Bob's problem
is now, given a vertex $v$ in the left part of one of the $G_i$, as
well as a subset $T$ of half of the vertices in the right part of that $G_i$,
decide if $|N(v) \cap T| > \frac{1}{4\eps^2} + \frac{c}{\eps}$, or if
$|N(v) \cap T| < \frac{1}{4\eps^2} - \frac{c}{\eps}$, for a small
constant $c > 0$. Most vertices $v$ will satisfy one of these
conditions, by anti-concentration of the binomial distribution. Note
that this problem is not a cut query problem, and so a priori it is
not clear how Bob can use Alice's sketch to solve it.

To solve the problem, Bob will do an exhaustive enumeration on cut
queries, and here is where we use that Alice's sketch preserves all
cuts. Namely, for each subset $S$ of half of the vertices in the left
part of $G_i$, Bob queries the cut $S \cup T$.  As Bob ranges over all
(exponentially many) such cuts, what will happen is that for most
vertices $u$ in the left part for which $|N(u) \cap T| >
\frac{1}{4\eps^2} + \frac{c}{\eps}$, the capacity of $S \cup T$ is a
``little bit'' larger if $u$ is excluded from $S$. This little bit is
not enough to be detected, since $|N(u) \cap T| = \Theta \left
(\frac{1}{\eps^2} \right )$ while the capacity of $S \cup T$ is
$\Theta \left ( \frac{1}{\eps^4} \right )$. However, as Bob range over
all such $S$, he will eventually get lucky in that $S$ contains all
vertices $u$ for which $|N(u) \cap T| > \frac{1}{4\eps^2} +
\frac{c}{\eps}$, and now since there are about $\frac{1}{2\eps^2}$
such vertices, the little $\frac{c}{\eps}$ bit gets ``amplified'' by a
factor of $\frac{1}{2\eps^2}$, which is just enough to be detected by
a $(1+\eps)$-approximation to the capacity of $S \cup T$. If Bob finds
the $S$ which maximizes the (approximate) cut value $S \cup T$, he can
check if his $v$ is in $S$, and this gives him a correct answer with
large constant probability. 

We believe our main contribution
is in designing a communication problem which requires Alice's sketch 
to preserve all cuts instead of only a single cut. 
There are also several details in the communication
lower bound for the problem itself, including a direct-sum theorem
for a constrained version of the Gap-Hamming-Distance problem,
which could be independently useful. 

For the $\Omega(n/\eps^2)$ edge lower bound for cut sparsifiers, 
the straightforward encoding would
encode each edge using $O(\log n)$ bits, and cause us to lose a $\log n$ factor
in the lower bound. Instead, we show how to randomly
round each edge weight in the sparsifier 
to an adjacent {\it integer}, and observe that the integer 
weights sum up to a small value in our communication problem. 
This ultimately allows to transmit, in a communication-efficient manner,
all the edge weights together with the edge identities.

Our lower bound for arbitrary data structures is given in Theorem \ref{thm:sketchLB}, 
and the application to cut sparsifiers is given in Theorem \ref{thm:sparsifierLB}.

\paragraph{Upper Bound.}
We now discuss the main ideas behind our sketch construction. Let us
first give some intuition why the previous algorithms cannot yield a
$\tilde O(n/\eps)$ bound, and show how our algorithm circumvents these
roadblocks on a couple of illustrative examples. For concreteness, it
is convenient to think of $\eps=1/\sqrt{n}$.

All existing cut and spectral sparsifiers construct the sparsifier
graph by taking a subgraph of the original graph $G$, with the
``right'' re-weightening of the edges \cite{BK96,SS11,BSS12,FHHP11,KP12}. 
In fact, except for \cite{BSS12}, they all proceed by sampling edges 
independently, each with its own probability (that depends on the graph).

Consider for illustration the complete graph. In this case, 
these sampling schemes employ a uniform probability $p\approx
\tfrac{1/\eps^2}{n}$ of sampling every edge. It is not hard to
see that one cannot sample edges with probability less than $p$, as
otherwise anti-concentration results suggest that even the degree of a
vertex (i.e., the cut of a ``singleton'') is not preserved within 
$1+\eps$ approximation. 
Perhaps a more interesting example is a random graph ${\mathcal G}_{n,1/2}$;
if edges are sampled independently with (roughly) uniform probability,
then again it cannot be less than $p$, because of singleton cuts.
However, if we aim for a sketch for the complete graph or ${\cal G}_{n,1/2}$,
we can just store the degree of each vertex using only $O(n)$
space, and this will allow us to report the value of every singleton
cut (which is the most interesting case, as the standard deviation for
these cut values have multiplicative order roughly $1\pm \eps$).
These observations suggest that {\em sketching} a graph may go beyond 
considering a subgraph (or a different graph) to represent the
original graph $G$.

Our general algorithm proceeds in several steps. The core of our
algorithm is a procedure for handling cuts of value $\approx 1/\eps^2$ 
in a graph with unweighted edges, which proceeds as follows. 
First, repeatedly partition the graph along every {\em sparse} cut, 
namely, any cut whose sparsity is below $1/\eps$. 
This results with a partition of the vertices into some number of parts.
We store the cross-edges (edge connecting different parts) explicitly.
We show the number of such edges is only $\tilde O(n/\eps)$, 
and hence they fit into the space alloted for the sketch.
Obviously, the contribution of these edges to any desired cut $w(S,\bar S)$ 
is easy to compute from this sketch.

The sketching algorithm still needs to estimates the contribution 
(to a cut $w(S,\bar S)$ for a yet unknown $S\subset V$)
from edges that are inside any single part $P$ of the partition. 
To accomplish this, we sample $\approx 1/\eps$ edges out of each vertex, 
and also store the exact degrees of all vertices. 
Then, to estimate the contribution of edges inside a part $P$ to $w(S,\bar S)$,
we take the sum of (exact) degrees of all vertices in $S\cap P$,
\emph{minus} an estimate for (twice) the number of edges inside $S\cap P$
(estimated from the edge sample).
This ``difference-based'' estimate has a smaller variance than 
a direct estimate for the number edges in $(S\cap P, \bar S\cap P)$ 
(which would be the ``standard estimate'', in some sense
employed by previous work). 
The smaller variance is achieved thanks to the facts that 
(1) the assumed cut is of size (at most) $1/\eps^2$; and 
(2) there are no sparse cuts in $P$.

Overall, we achieve a sketch size of $\tilde O(n/\eps)$.  We can
construct the sketch in polynomial time by employing an $O(\sqrt{\log
  n})$-approximation algorithm for sparse cut
\cite{ARV09,sherman2009breaking} or faster algorithms with
$(\log^{O(1)} n)$-approximation \cite{madry2010fast}.

The formal statement of our main upper bound appears in 
Theorem~\ref{thm:upper}.

\section{Lower Bounds}
\label{sec:LB}

We first prove a lower bound on the size of ``for all'' sketches,
and then use it to prove a tight lower bound on the size of cut sparsifiers.
(For lower bounds on the size of ``for each'' sketches, 
see Section~\ref{sec:nEpsLB}.)

\begin{theorem} \label{thm:sketchLB}
Fix an integer $n$ and $\veps\in [1/n,1/{\sqrt n}]$, 
and let $\sk=\sk_{n,\veps}$ and $\est=\est_{n,\veps}$ be ``sketching'' 
and ``estimation'' algorithms for (unweighted) graphs on vertex set $[n]$,
which may be randomized.
Suppose that for every such graph $G=([n],E)$, 
with probability at least $3/4$ the resulting sketch $\sk(G)$ satisfies%
\footnote{The probability is over the two algorithms' common randomness;
equivalently, the two algorithms have access to a public source of random bits}
$$
  \forall S\subset [n],\quad \est\big(S,\sk(G)\big) \in\ (1\pm\veps)\cdot \card{E(S,\bar S)}.
$$
Then the worst-case space requirement of $\sk$ is $\Omega(n/\veps^2)$ bits.
\end{theorem}

\begin{theorem} \label{thm:sparsifierLB}
For every integer $n$ and $\veps\in [1/n,1/{\sqrt n}]$, 
there is an $n$-vertex graph $G$ for which every $(1+\eps)$-cut sparsifier $H$
has $\Omega(n/\eps^2)$ edges, 
even if $H$ is not required to be a subgraph of $G$.
\end{theorem}

\subsection{Sketch-size lower bound}
\label{sec:sketchLB}

We prove Theorem~\ref{thm:sketchLB} using the following communication 
lower bound for a version of the Gap-Hamming-Distance problem, 
whose proof appears in Section~\ref{sec:CC}.
Throughout, we fix $c\eqdef 10^{-3}$ (or a smaller positive constant),
and assume $\veps\leq c/10$.

\begin{theorem} \label{thm:GHD}
Consider the following distributional communication problem:
Alice has as input $n/2$ strings $s_1,\ldots,s_{n/2} \in \zo^{1/\veps^2}$
of Hamming weight $\tfrac{1}{2\veps^2}$, 
and Bob has an index $i\in[n/2]$ together with one string $t\in \zo^{1/\veps^2}$
of Hamming weight $\tfrac{1}{2\veps^2}$, 
drawn as follows:%
\footnote{Alice's input and Bob's input are \emph{not} independent,
but the marginal distribution of each one is uniform over its domain,
namely, $\zo^{(n/2)\times (1/\veps^2)}$ and $[n]\times \zo^{1/\veps^2}$,
respectively.
}
\begin{itemize} \compactify
\item $i$ is chosen uniform at random; 
\item $s_i$ and $t$ are chosen uniformly at random but conditioned on 
their Hamming distance $\Delta(s_i,t)$ being, with equal probability,  
either $\geq \tfrac{1}{2\veps^2} + \tfrac{c}{\veps}$ 
or     $\leq \tfrac{1}{2\veps^2} - \tfrac{c}{\veps}$;
\item the remaining strings $s_{i'}$ for $i'\neq i$ are chosen uniformly at random.
\end{itemize}
Consider a (possibly randomized) one-way protocol,
in which Alice sends to Bob an $m$-bit message,
and then Bob determines, with success probability at least $2/3$, 
whether $\Delta(s_i,t)$
is $\geq \tfrac{1}{2\veps^2} + \tfrac{c}{\veps}$ or $\leq \tfrac{1}{2\veps^2} - \tfrac{c}{\veps}$.
Then Alice's message size is $m\geq \Omega(n/\veps^2)$ bits. 
\end{theorem}

We can interpret the lower bound of Theorem \ref{thm:GHD} as follows:
Consider a (possibly randomized) algorithm that produces an $m$-bit
sketch of Alice's input $(s_1,\ldots,s_{n/2})\in\zo^{n/2\veps^2}$, and
suppose that the promise about $\Delta(s_i,t)$ can be decided
correctly (with probability at least $3/4$) given (only) the sketch
and Bob's input $(i,t)\in [n/2]\times \zo^{1/\veps^2}$.  Then $m\geq
\Omega(n/\veps^2)$.

We now prove Theorem \ref{thm:sketchLB} by a reduction to the above communication problem,
interpreting the one-way protocol as a sketching algorithm, as follows.
Given the instance $(s_1,\ldots,s_{n/2},i,t)$, define an $n$-vertex graph $G$ 
that is a disjoint union of the graphs $\aset{G_j: j\in[\veps^2n/2]}$,
where each $G_j$ is a bipartite graph, whose two sides,
denoted $L(G_j)$ and $R(G_j)$, are of size 
\[ \card{L(G_j)}=\card{R(G_j)}=1/\veps^2. \]
The edges of $G$ are determined by $s_1,\ldots,s_{n/2}$,
where each string $s_u$ is interpreted as a vector of indicators for 
the adjacency between vertex $u\in \cup_{j\in[\veps^2n/2]} L(G_j)$
and the respective $R(G_j)$.

Observe that Alice can compute $G$ without any communication,
as this graph is completely determined by her input.
She then builds a sketch of this graph, that with probability $\geq 99/100$, 
succeeds in simultaneously approximating all cut queries within factor 
$1\pm\gamma\veps$, where $\gamma>0$ is a small constant to be determined later. 
This sketch is obtained from the theorem's assumption about $m$-bit sketches 
by standard amplification of the success probability from $3/4$ to $0.99$
(namely, repeating $r=O(1)$ times independently 
and answering any query with the median value of the $r$ answers).
Alice then sends this $O(m)$-bit sketch to Bob. 

Bob then uses his input $i$ to compute $j=j(i)\in[\veps^2n/2]$ 
such that the graph $G_j$ contains vertex $i$ 
(i.e., the vertex whose neighbors are determined by $s_i$).
Bob also interprets his input string $t$ as a vector of indicators
determining a subset $T\subseteq R(G_j)$.
By construction of $G_j$ 
, the neighbor sets $N(v)$ of the vertices $v\in L(G_j)\setminus \aset{i}$
are uniformly distributed, independently of $T$ and of each other;
in particular, each $\card{N(v)\cap T}$ has a Binomial distribution 
$B(\tfrac{1}{\veps^2},\tfrac14)$.

\begin{lemma} \label{lem:BobRecovers}
Using the $O(m)$-bit sketch he received from Alice, 
Bob can compute a ``list'' $B\subset L(G_j)$ of size 
$\card{B} = \half\card{L(G_j)} = \tfrac{1}{2\veps^2}$,
and with probability at least $0.96$, this list contains 
at least $\tfrac{4}{5}$-fraction of the vertices in the set 
\begin{equation} \label{eq:Nlarge}
 \Lhigh \eqdef \aset{v\in L(G_j):\ 
      \card{N(v)\cap T} \geq \tfrac{1}{4\veps^2} + \tfrac{c}{\veps}
      }.
\end{equation}
Moreover, Bob uses no information about his input $i$ other than $j=j(i)$.
\end{lemma}

Before proving the lemma, let us show how to use it to decide 
about $\Delta(s_i,t)$ and derive the theorem.
We will need also the following simple claim, 
which we prove further below. %

\begin{claim} \label{cl:binom1}
With probability at least $0.98$, 
the relative size of $\Lhigh$ is 
$
  \frac{ \card{\Lhigh} }{ \card{L(G_j)} } \in [\half-10c,\half]
$.
\end{claim}

We assume henceforth that the events described in the above lemma and claim
indeed occur, which happens with probability at least $0.94$.
Notice that $\Delta(s_i,t) = \deg(i) + \card{T} - 2\card{N(i)\cap T}$.

Now suppose that $\Delta(s_i,t) \leq \tfrac{1}{2\veps^2} - \tfrac{c}{\veps}$.
Then $\card{N(i)\cap T} \geq \tfrac{1}{4\veps^2} + \tfrac{c}{2\veps}$, 
and because Bob's list $B$ is independent of the vertex $i\in L(G_j)$,
we have 
$\Pr[i\in B] 
 \geq \tfrac{4}{5}\card{\Lhigh} / \card{\Lhigh}
 = \tfrac{4}{5}
$.

Next, suppose that $\Delta(s_i,t) \geq \tfrac{1}{2\veps^2} + \tfrac{c}{\veps}$.
Then $\card{N(i)\cap T} \leq \tfrac{1}{4\veps^2} - \tfrac{c}{2\veps}$, 
and using Claim \ref{cl:binom1},
\[
  \Pr[i\in B] 
  \leq \frac{ \card{B} - \tfrac{4}{5}\card{\Lhigh} }{ \card{L(G_j)} } 
  \leq \frac{1}{4}.
\]
Thus, Bob can decide between the two cases 
with error probability at most $1/4$. 
Overall, it follows that Bob can solve the Gap-Hamming-Distance problem for $(s_i,t)$,
with overall error probability at most $1/4 + 0.06 < 1/3$,
as required to prove the theorem.

\bigskip
\begin{proof}[Proof of Claim \ref{cl:binom1}]
By basic properties of the Binomial distribution (or the Berry-Esseen Theorem),
there are absolute constants $\tfrac{1}{5}\leq K_1 \leq K_2\leq 5$ 
such that for each vertex $v\in L(G_j)$, 
\[
  \Pr\Big[v\in \Lhigh\Big] 
  = \Pr\Big[\card{N(v)\cap T} \geq \tfrac{1}{4\veps^2} + \tfrac{c}{\veps}\Big] 
  \in [\half - K_2 c, \half - K_1 c].
\]
Denoting $Z\eqdef \card{\Lhigh}$, we have by Hoeffding's inequality, 
\[
  \Pr\Big[ \abs{Z-\EX[Z]} > \tfrac{K_1 c}{\veps^2} \Big]
  \leq 2e^{-\half (K_1 c)^2 (1/\veps^2)}
  \leq 0.02.
\]
Thus, with probability at least $0.98$, we have both bounds
\begin{align*}
  Z &\leq \EX[Z] + \tfrac{K_1 c}{\veps^2} \leq \tfrac{1}{2\veps^2}, 
  \text{ and } \\
  Z &\geq \EX[Z] - \tfrac{K_1 c}{\veps^2} \leq \tfrac{1}{2\veps^2} 
     \geq (\half-2K_2c)\tfrac{1}{\veps^2} 
     \geq \tfrac{1}{4\veps^2}.
\end{align*}
\end{proof}

\bigskip
\begin{proof}[Proof of Lemma~\ref{lem:BobRecovers}]
We now show how Bob creates the ``list'' $B\subset L(G_j)$ of size 
$\card{B}=\tfrac{1}{2\veps^2}$. 
Bob estimates the cut value for $S\cup T$ 
for every subset $S\subseteq L(G_j)$ of size exactly $\frac{1}{2\veps^2}$. 
Observe that the cut value for a given $S$ is 
\[
  \delta(S\cup T)
  = \sum_{v\in S} \deg(v) + \sum_{u\in T} \deg(u) - 2 \sum_{v\in S} \card{N(v)\cap T}.
\]

The cut value is bounded by the number of edges in $G_j$, which is at most $1/\veps^4$, 
and since the data structure maintains all the cut values 
within factor $1+\gamma\veps$ for an arbitrarily small constant $\gamma>0$, 
the additive error on each cut value is at most $\gamma/\veps^3$
Further, we can assume Bob knows the exact degrees of all vertices 
(by adding them to the sketch, using $O(n\log\tfrac{1}{\veps})$ bits), 
which he can subtract off, and since scaling by $-1/2$ can only shrink the
additive error, we can define the ``normalized'' cut value
\[ 
  n(S,T) \eqdef \sum_{v\in S} \card{N(v)\cap T}, 
\]
which Bob can estimate within additive error $\gamma/\veps^3$. 
Bob's algorithm is to compute these estimates for all the values $n(S,T)$,
and output a set $S$ that maximizes his estimate for $n(S,T)$ 
as the desired list $B\subset L(G_j)$.

Let us now analyze the success probability of Bob's algorithm.
For each vertex $v\in L(G_j)$, let $f(v) \eqdef \card{N(v)\cap T}$. 
Observe that each $f(v)$ has a Binomial distribution 
$B(\tfrac{1}{\veps^2},\tfrac14)$, and they are independent of each other.
We will need the following bounds on the typical values of some order statistics
when taking multiple samples from such a Binomial distribution.
Recall that the \emph{$k$-th order statistic} of a sample (multiset) 
$x_1,\ldots,x_m\in\R$ is the $k$-th smallest element in that sample.
The following claim is proved further below. %

\begin{claim} \label{cl:binom}
Let $\aset{X_j}_{j=1,\ldots,m}$ be independent random variables
with Binomial distribution $B(t,\tfrac14)$.
Let $\alpha\in(0,\half)$ 
such that $(\half+\alpha)m$ is integral,
and both $t,m \geq 10/\alpha^2$.
Then
\begin{align*}
  \Pr&\Big[\text{the $(\half-\alpha)m$ order statistic of $\aset{X_j}$ 
is $\leq \tfrac14 t - \tfrac{\alpha}{10}\sqrt{t}$}\Big] \geq 0.99,
  \text{ and } \\
  \Pr&\Big[\text{the $(\half+\alpha)m$ order statistic of $\aset{X_j}$ 
is $\geq \tfrac14 t + \tfrac{\alpha}{10}\sqrt{t}$}\Big] \geq 0.99.
\end{align*}
\end{claim}

Sort the vertices $v\in L(G_j)$ by their $f(v)$ value, 
and denote them by $v_1,\ldots,v_{1/\veps^2}$ such that $f(v_i) \leq f(v_{i+1})$.
Applying the claim (for $\alpha=0.05$ and $t,m=\tfrac{1}{\veps^2}$), 
we see that with probability at least $0.98$, the difference 
\begin{equation} \label{eq:fdiff}
  f(v_{0.55/\veps^2})-f(v_{0.45/\veps^2}) 
  \geq 0.01 / {\veps}.
\end{equation}
We assume henceforth this event indeed occurs.
Let $S^*$ include the $\frac{1}{2\veps^2}$ vertices $v\in L(G_j)$ 
with largest $f(v)$, i.e., $S^*\eqdef \aset{v_j}_{j > 0.5/\veps^2}$,
and let $S'\subset L(G_j)$ be any subset of the same size such that 
at least $\tfrac{1}{10}$-fraction of its vertices are not included in $S^*$
(i.e., their order statistic in $L(G_j)$ is at most $\tfrac{1}{2\veps^2}$).
Then we can write
\begin{align*}
  n(S^*,T)
  &= \sum_{j\in S^*} f(v) 
   = \sum_{j>0.5/\veps^2} f(v_j),
  \\
  n(S',T) 
  &= \sum_{j\in S'} f(v) 
   \leq \sum_{j>0.6/\veps^2} f(v_j) + \sum_{0.4/\veps^2 < j\leq 0.5/\veps^2} f(v_j).
\end{align*}
Now subtract them
\begin{align*}
  n(S',T) - n(S^*,T)
  &= \sum_{0.5/\veps^2 < j\leq 0.6/\veps^2} f(v_j)
    - \sum_{0.4/\veps^2 < j\leq 0.5/\veps^2} f(v_j),
\intertext{observe that elements in the normalized interval $(0.5,0.55]$ 
dominate those in $(0.45,0.5]$,
}
  &\geq \sum_{0.55/\veps^2 < j\leq 0.6/\veps^2} f(v_j)
    - \sum_{0.4/\veps^2 < j\leq 0.45/\veps^2} f(v_j)
\intertext{and bound the remaining elements using \eqref{eq:fdiff}, }
  &\geq (0.05/\veps^2)  \big[ f(v_{0.55/\veps^2})-f(v_{0.45/\veps^2}) \big]
  \geq 0.0005/\veps^3.
\end{align*}
Bob's estimate for each of the values $n(S^*,T)$ and $n(S',T)$ 
has additive error at most $\gamma/\veps^3$, 
and therefore for suitable $\gamma=10^{-4}$, 
the list $B$ Bob computes cannot be this set $S'$.
Thus, Bob's list $B$ must contain at least $9/10$-fraction of $S^*$,
i.e., the $\frac{1}{2\veps^2}$ vertices $v\in L(G_j)$ with highest $f(v)$.

Recall from Claim~\ref{cl:binom1} that with probability at least $0.98$, 
we have $\tfrac{1}{4\veps^2} \leq \card{\Lhigh} \leq \tfrac{1}{2\veps^2}$,
and assume henceforth this event occurs.
Since $S^*$ includes the $\tfrac{1}{2\veps^2}$ vertices with highest $f$-value,
it must contain all the vertices of $\Lhigh$,
i.e., $\Lhigh \subseteq S^*$.
We already argued that Bob's list $B$ contains all 
but at most $\tfrac{1}{10}\card{S^*} = \tfrac{1}{20\veps^2}$ vertices of $S^*$,
and thus 
\[
  \frac{\card{\Lhigh\setminus B}} {\card{\Lhigh}}
  \leq \frac{\card{S^*\setminus B}} {\card{\Lhigh}}
  \leq \frac{\ \tfrac{1}{20\veps^2}\ } {\tfrac{1}{4\veps^2} }
  = \frac{1}{5}.
\]
This bound holds with probability at least $0.96$
(because of two events that we ignored, each having probability at most $0.02$)
and this proves Lemma~\ref{lem:BobRecovers}.
\end{proof}

\bigskip
\begin{proof}[Proof of Claim \ref{cl:binom}]
The $(\half-\alpha)m$ order statistic of $\aset{X_j}$ 
is smaller or equal to $T\eqdef \tfrac14 t - \tfrac{\alpha}{10}\sqrt{t}$
if and only if at least $(\half-\alpha)m$ elements are smaller or equal to $T$, 
which can be written as $\sum_{j=1}^m \indic_{\aset{X_j\le T}} \geq (\half-\alpha)m$.

Now fix $j\in\aset{1,\ldots,t}$. 
Then 
\begin{equation} \label{eq:binom1}
  \Pr[X_j\leq T] = \Pr[X_j \leq \tfrac14 t] \cdot \Pr[X_j \leq T \mid X_j \leq \tfrac14 t],
\end{equation}
and by the Binomial distribution's relationship between mean and median,
$\Pr[X_j \leq \tfrac14 t] \geq \half$.
Elementary but tedious calculations (or the Berry-Esseen Theorem) show 
there is an absolute constant $K\in(0,5)$ such that 
\[
  \Pr\Big[\tfrac14 t - \tfrac{\alpha}{10}\sqrt{t} < X_j \leq \tfrac14 t\Big]
  \leq K \tfrac{\alpha}{10}\cdot \Pr\Big[X_j \leq \tfrac14 t\Big],
\]
and plugging into \eqref{eq:binom1}, we obtain 
$\Pr[X_j\leq T] \geq \half(1-K\tfrac{\alpha}{10}) \geq \half - \half\alpha$.

Now bound the expectation by
$\EX[\sum_{j=1}^m \indic_{\aset{X_j\le T}}] \geq (\half -\half\alpha)m$,
and apply Hoeffding's inequality, 
\[
  \Pr \Big[\sum_j \indic_{\aset{X_j\le T}} < (\half-\alpha)m \Big]
  \leq e^{-\half (\half \alpha)^2 m}
  = e^{-\alpha^2 m/8}
  \leq 0.01,
\]
where the last inequality follows since $\alpha^2m$ is sufficiently large.
\end{proof}

\subsection{The communication lower bound}
\label{sec:CC}

We now prove Theorem \ref{thm:GHD} (see Theorem \ref{thm:GHD2} below),
by considering distributional communication problems
between two parties, Alice and Bob, as defined below. 
We restrict attention to the one-way model, in which Alice sends to Bob
a single message $M$ that is a randomized function of her input
(using her private randomness), and Bob outputs the answer.

\paragraph{Distributional versions of Gap-Hamming-Distance.}
Recall that our analysis is asymptotic for $\veps$ tending to $0$, and 
let $0 < c < 1$ be a parameter, considered to be a sufficiently small constant.
Alice's input is $S \in \{0,1\}^{\frac{1}{\eps^2}}$,
Bob's input is $T \in \{0,1\}^{\frac{1}{\eps^2}}$,
where the Hamming weights are $\wt(S) = \wt(T) = \frac{1}{2\eps^2}$, and Bob 
needs to evaluate the partial function
\[
  f_{c}(S,T) = 
  \begin{cases}
    1& \text{if $\Delta(S,T) \geq \frac{1}{2\eps^2} + \frac{c}{\eps}$; }\\
    0& \text{if $\Delta(S,T) \leq \frac{1}{2\eps^2} - \frac{c}{\eps}$. }
  \end{cases}
\]
The distribution $\mu$ we place on the inputs $(S,T)$ is the following: 
$S$ is chosen uniformly at random with $\wt(S) = \frac{1}{2\eps^2}$,
and then with probability $\frac{1}{2}$, we choose $T$ uniformly at random with $\wt(T) = \frac{1}{2\eps^2}$
subject to the constraint that $\Delta(S,T) \geq \frac{1}{2\eps^2} + \frac{c}{\eps}$, 
while with the remaining probability $\frac{1}{2}$, we choose $T$
uniformly at random with $\wt(T) = \frac{1}{2\eps^2}$ 
subject to the constraint that $\Delta(S,T) \leq \frac{1}{2\eps^2} - \frac{c}{\eps}$. 
We say Alice's message $M = M(S)$ is $\delta$-error for $(f_{c}, \mu)$ if 
Bob has a reconstruction function $R$ for which
$$\Pr_{(S,T) \sim \mu, \textrm{ private randomness}}[R(M, T) = f_{c}(S,T)] \geq 1-\delta.$$

Now consider a related but different distributional problem. 
Alice and Bob have $S, T \in \{0,1\}^{\frac{1}{\eps^2}}$, respectively,
each of Hamming weight exactly $\frac{1}{2\eps^2}$, 
and Bob needs to evaluate the function
\[
  g(S,T) = 
  \begin{cases}
    1& \text{if $\Delta(S,T) >    \frac{1}{2\eps^2}$; }\\
    0& \text{if $\Delta(S,T) \leq \frac{1}{2\eps^2}$. }
  \end{cases}
\]
We place the following distribution $\zeta$ on the inputs $(S,T)$: 
$S$ and $T$ are chosen independently and uniformly at random
among all vectors with Hamming weight exactly $\frac{1}{2\eps^2}$. 
We say a message $M$ is $\delta$-error for $(g, \zeta)$ if Bob has a reconstruction function $R$ for which
$$\Pr_{(S, T) \sim \zeta, \ \textrm{private randomness}}[R(M, T) = g(S,T)] \geq 1-\delta.$$

Let $I(S ; M) = H(S) - H(S | M)$ be the mutual information between $S$ and $M$, where $H$ is the entropy
function. Define
$\IC_{\mu, \delta}(f_c) \eqdef \min_{\text{$\delta$-error $M$ for $(f_c, \mu)$}} I(S ; M)$ 
and 
$\IC_{\zeta, \delta}(g) \eqdef \min_{\text{$\delta$-error $M$ for $(g, \zeta)$}} I(S ; M)$. 

\begin{lemma}\label{lem:icost}
For all $\delta>0$, 
 $\IC_{\mu, \delta}(f_c) 
\geq \IC_{\zeta, \delta + O(c)}(g)$. 
\end{lemma}
\begin{proof}
It suffices to show that if 
$M$ is $\delta$-error for $(f_c, \mu)$, then $M$ is $(\delta + O(c))$-error for $(g, \zeta)$.
Since $M$ is $\delta$-error for $(f_c, \mu)$, Bob has a reconstruction function $R$ for which
$$\Pr_{(S,T) \sim \mu, \textrm{ private randomness}}[R(M, T) = f_{c}(S,T)] \geq 1-\delta.$$ Now consider
$\Pr_{(S,T) \sim \zeta, \textrm{ private randomness}}[R(M, T) = g(S,T)]$. Observe that whenever $(S,T)$ lies in the support of $\mu$, if
$R(M,T) = f_c(S,T)$, then $R(M,T) = g(S,T)$. The probability that $(S,T)$ lies in the support of $\mu$
is $1-O(c)$, by standard anti-concentration arguments of the Binomial distribution (or the Berry-Esseen Theorem), 
and conditioned on this event we have that $(S,T)$ is distributed according to $\mu$. Hence,
$\Pr_{(S,T) \sim \zeta, \ \textrm{private randomness}}[R(M, T) = g(S,T)] 
  \geq [1 - O(c)][1- \delta]
  \geq 1 - O(c) - \delta
  $. 
\end{proof}
We now lower bound $\IC_{\zeta, \delta}(g)$. 
\begin{lemma}\label{lem:icost2}
For $\delta_0 > 0$ a sufficiently small constant, $IC_{\zeta, \delta_0}(g) = \Omega \left (\frac{1}{\eps^2} \right ).$
\end{lemma}
\begin{proof}
We use the following lower bound of 
Braverman, Garg, Pankratov and Weinstein \cite{BGPW13} 
for the following $h_c(S,T)$ problem. Like before, Alice
has $S \in \{0,1\}^{\frac{1}{\eps^2}}$, Bob has $T \in \{0,1\}^{\frac{1}{\eps^2}}$, 
and needs to evaluate the partial function
\[
  h_{c}(S,T) = 
  \begin{cases}
    1& \text{if $\Delta(S,T) \geq \frac{1}{2\eps^2} + \frac{c}{\eps}$; }\\
    0& \text{if $\Delta(S,T) \leq \frac{1}{2\eps^2} - \frac{c}{\eps}$. }
  \end{cases}
\]
However, now $\wt(S)$ and $\wt(T)$ may be arbitrary. 
Moreover, $S$ and $T$ are chosen independently and uniformly at random from $\{0,1\}^{\frac{1}{\eps^2}}$. Denote
this by $(S,T) \sim \eta$. 
Now it may be the case that $\left|\Delta(S,T) - \frac{1}{2\eps^2}\right| < \frac{c}{\eps}$,
in which case Bob's output is allowed to be arbitrary. A message $M$ is $\delta$-error for $(h_c, \eta)$ 
if Bob has a reconstruction
function $R$ for which
$$\Pr_{(S,T) \sim \eta, \textrm{ private randomness}} \left [\left (R(M,T) = h_c(S,T) \right )\wedge \left (\left| \Delta(S,T) - \frac{1}{2\eps^2} \right| \geq \frac{c}{\eps} \right )\right ] \geq 1-\delta.$$
It was proved in \cite{BGPW13} that for a sufficiently small constant $\delta > 0$, 
$$\IC_{\eta, \delta}(h_1) \eqdef \min_{\text{$\delta$-error $M$ for $(h_1,\eta)$}} I(S ; M) \geq \frac{C}{\eps^2},$$ for an absolute
constant $C > 0$. We show how to apply this result to prove the lemma.

An immediate corollary of this result is that 
$\IC_{\eta, \delta}(g) \eqdef \min_{\text{$\delta$-error $M$ for $(g,\eta)$}} I(S ; M) \geq \frac{C}{\eps^2}$. 
Indeed, if $M$ is $\delta$-error for $(g, \eta)$, then it is also $\delta$-error
for $(h_1, \eta)$. 

Now let $M$ be a $\delta$-error protocol for $(g, \zeta)$. Consider the following randomized protocol $M'$
for $g$ with inputs distributed according to $\eta$. 
Given $S$, Alice computes $s =$ wt($S$). If $s < \frac{1}{2\eps^2}$, Alice randomly chooses
$\frac{1}{2\eps^2} - s$ coordinates in $S$ that are equal to $0$ and replaces them with a $1$, otherwise she randomly
chooses $s - \frac{1}{2\eps^2}$ coordinates in $S$ that are equal to $1$ and replaces them with a $0$. Let $S'$ be
the resulting vector. Alice sends $M(S')$ to Bob, i.e., $M'(S)\eqdef M(S')$. 
Given the message $M(S')$ and his input $T$, Bob first computes $t = \wt(T)$. 
If $t < \frac{1}{2\eps^2}$,
Bob randomly chooses $\frac{1}{2\eps^2} - t$ coordinates in $T$ which are equal to $0$ and replaces them with a $1$, otherwise
he randomly chooses $t-\frac{1}{2\eps^2}$ coordinates in $T$ which are equal to $1$ and replaces them with a $0$. Let $T'$
be the resulting vector. Suppose $R$ is such that 
$\Pr_{(S', T') \sim \zeta, \textrm{ private randomness}} [R(M(S'), T') = g(S', T')] \geq 1-\delta.$ Bob outputs $R(M(S'),T')$.

We now lower bound $\Pr[g(S', T') = g(S, T)]$, where the probability is over $(S,T) \sim \eta$ and the random choices of Alice
and Bob for creating $S', T'$ from $S,T$, respectively. First, the number of coordinates changed by Alice or Bob
is $r = \Theta(1/\eps)$ with arbitrarily large constant probability. Since $S$ and $T$ are independent and uniformly random, 
after performing this change, the Hamming distance on these $r$ coordinates is $\frac{r}{2} \pm O(\sqrt{r})$ with
arbitrarily large constant probability. Finally, $\left|\Delta(S', T') - \frac{1}{2\eps^2} \right| = \omega(\sqrt{r})$
with arbitrarily large constant probability. Hence, with arbitrarily large constant probability, $g(S',T') = g(S,T)$.
It follows that $\Pr[g(S', T') = g(S, T)] \geq 1-\gamma$ for an arbitrarily small constant $\gamma > 0$, and
therefore if $R'$ describes the above reconstruction procedure of Bob, then
$\Pr_{(S, T) \sim \eta, \textrm{ private randomness}}[R'(M'(S),T) = g(S, T)] \geq 1-\gamma - \delta$. 

Hence, $M'$ is a $(\delta+\gamma)$-error protocol for $(g, \eta)$. We now bound $I(M' ; S)$ in terms of $I(M ; S')$. 
Let $J$ be an indicator random variable for the event 
$\wt(S)\in \left [\frac{1}{2\eps^2}-\frac{1}{\eps^{3/2}}, \frac{1}{2\eps^2} + \frac{1}{\eps^{3/2}} \right ]$. 
Then $\Pr[J = 1] = 1-o(1)$,
where $o(1) \rightarrow 0$ as $\eps \rightarrow 0$. Since conditioning on a random variable $Z$ can change the mutual
information by at most $H(Z)$, we have 
\begin{eqnarray}\label{eqn:derive1}
I(M' ; S) \leq I(M' ; S \mid J) + H(J) \leq I(M' ; S \mid J = 1) + 1.
\end{eqnarray}
$S$ is a probabilistic function of $S'$, which if $J = 1$, is obtained by changing at most $1/\eps^{3/2}$ randomly
chosen coordinates $A_1, \ldots, A_{1/\eps^{3/2}}$ of $S'$ from $0$ to $1$ or from $1$ to $0$. By the data processing inequality
and the chain rule for mutual information, 
\begin{eqnarray}\label{eqn:derive2}
I(M' ; S \mid J = 1) & \leq & I(M' ; S', A_1, \ldots, A_{1/\eps^{3/2}} \mid J = 1) \nonumber \\
& = & I(M' ; S' \mid J = 1) + \sum_{\ell=1}^{1/\eps^{3/2}} I(M' ; A_{\ell} \mid J = 1, A_1, \ldots, A_{\ell-1}) \nonumber \\
& \leq & I(M' ; S' \mid J = 1) + O \left (\frac{\log(1/\eps)}{\eps^{3/2}} \right ).
\end{eqnarray}
Observe that the joint distribution of $M'(S')$ and $S'$ is independent of $J$, and moreover is equal to the joint distribution
of $M(S')$ and $S' \sim \zeta$. We can take $M$ to be a $\delta$-error protocol for $(g, \zeta)$ for which
$I(M(S') ; S') = \IC_{\zeta, \delta}(g)$. Combining this with (\ref{eqn:derive1}) and (\ref{eqn:derive2}), 
$I(M' ; S ) \leq \IC_{\zeta, \delta}(g) + O \left (\frac{\log(1/\eps)}{\eps^{3/2}} \right ).$
Now since $M'$ is a $(\delta+\gamma)$-error protocol for $(g, \eta)$, we have $I(M' ; S) \geq \IC_{\eta, \delta+\gamma}(g) \geq  \frac{C}{\eps^2}$,
provided $\delta$ and $\gamma$ are sufficiently small constants. It follows that
$ \IC_{\zeta, \delta}(G) 
  \geq \frac{C}{\eps^2} - O \left (\frac{\log(1/\eps)}{\eps^{3/2}} \right )
  \geq \frac{C}{2\eps^2}
$, as desired. 
\end{proof}
\begin{corollary}\label{cor:prim}
For sufficiently small constants $\delta, c > 0$,  $\IC_{\mu, \delta}(f_c) = \Omega(1/\eps^2)$. 
\end{corollary}
\begin{proof}
This follows by combining Lemmas \ref{lem:icost} and \ref{lem:icost2}. 
\end{proof}

\paragraph{$n$-fold version of Gap-Hamming-Distance.}
We now consider the $n$-fold problem in which Alice is given $n$ strings
$S_1, \ldots, S_n \in \{0,1\}^{1/\eps^2}$, and Bob has an index
$I \in [n]$ together with one string $T \in \{0,1\}^{1/\eps^2}$. Here $(S_I, T) \sim \zeta$, while $S_j$ for $j \neq I$, are
chosen independently and uniformly at random from all Hamming weight $\frac{1}{2\eps^2}$ vectors. 
Thus the joint distribution of $S_1, \ldots, S_n$ 
is $n$ i.i.d. strings drawn uniformly from $\{0,1\}^{1/\eps^2}$ subject to each of their Hamming weights being $\frac{1}{2\eps^2}$. 
Here $I$ is drawn independently and uniformly from $[n]$. 
We let $\nu$ denote the resulting input distribution. 

We consider the one-way two-party model in which Alice sends a single, possibly randomized message $M$ of her inputs
$S_1, \ldots, S_n$, and Bob needs to evaluate $h(S_1, \ldots, S_n, T) = f_c(S_I, T)$.
We say $M$ is $\delta$-error for $(h, \nu)$ if Bob has a reconstruction
function $R$ for which 
$$\Pr_{\textrm{inputs}\sim \nu, \textrm{ private randomness}} \left [\left (R(M,T, I) = f_c(S_I, T) \right )\wedge \left (\left| \Delta(S_I,T) - \frac{1}{2\eps^2} \right| \geq \frac{c}{\eps} \right )\right ] \geq 1-\delta.$$
Let $\IC_{\nu, \delta}(h) \eqdef \min_{\text{$\delta$-error $M$ for $(h, \nu)$}} I(S_1, \ldots, S_n ; M)$.

\begin{theorem} \label{thm:GHD2}
For a sufficiently small constant $\delta > 0$, $\IC_{\nu, \delta}(h) = \Omega(n/\eps^2)$. In particular, the distributional 
one-way communication complexity of $h$ under input distribution $\nu$ is $\Omega(n/\eps^2)$. 
\end{theorem}
\begin{proof}
Say an index $i \in [n]$ is {\it good} if
$$\Pr_{\textrm{inputs}\sim \nu, \textrm{ private randomness}} \left [\left (R(M,T, I) = f_c(S_I, T) \right )\wedge \left (\left| \Delta(S_I,T) - \frac{c}{2\eps^2} \right| \geq \frac{1}{\eps} \right ) \mid I = i\right ] \geq 1-2\delta.$$
By a union bound, there are at least $n/2$ good $i \in [n]$. By the chain rule for mutual information and 
using that the $S_i$ are independent and conditioning does not increase entropy, 
\begin{eqnarray*}
I(M ; S_1, \ldots, S_n) \geq \sum_{i=1}^n I(M ; S_i) \geq \sum_{\textrm{good }i} I(M ; S_i).
\end{eqnarray*}
We claim that for each good $i$, $I(M ; S_i) \geq  \IC_{\mu, 2\delta}(f_c)$. Indeed, consider the following protocol $M_i$ for $f_c$ under distribution $\mu$. 
Alice, given her input $S$ for $f_c$, uses her private randomness to sample $S_j$ for all $j \neq i$ independently and uniformly
at random from $\{0,1\}^{1/\eps^2}$ subject to each of their Hamming weights being $\frac{1}{2\eps^2}$. Bob sets $I = i$ and
uses his input $T$ for $f_c$ as his input for $h$. Since $i$ is good, it follows that $M_i$ is $2\delta$-error
for $(f_c, \zeta)$. Hence $I(M ; S_i) \geq \IC_{\mu, 2\delta}(f_c)$, which by Corollary \ref{cor:prim}, is $\Omega(1/\eps^2)$
provided $\delta > 0$ is a sufficiently small constant. Hence, $\IC_{\nu, \delta} = \Omega(n/\eps^2)$. Since 
$\IC_{\nu, \delta}(h) \leq I(M ; S_1, \ldots, S_n)$ for each $\delta$-error $M$ for $(h, \nu)$, and 
$I(M ; S_1, \ldots, S_n) \leq H(M)$ which is less than the length of $M$, the communication complexity lower bound follows. 
\end{proof}

\subsection{Cut-sparsifiers lower bound}
\label{sec:sparsifierLB}

We now prove Theorem~\ref{thm:sparsifierLB}.
Fix $n$ and $\eps$, and assume that every $n$-vertex graph has a $(1+\gamma\eps/2)$-cut sparsifier with at most $s$ edges, 
for a small constant $\gamma>0$ to be determined later. 
We wish to prove a lower bound on $s$.
Consider then the random graph $G$ constructed 
in the proof of Theorem~\ref{thm:sketchLB},
as a disjoint union of graphs $\aset{G_j: j\in[\veps^2n/2]}$, 
each being a bipartite graph with $1/\eps^2$ vertices in each side.
By our assumption above, such $G$ (always) has a subgraph $H$ 
which is a $(1+\gamma\eps/2)$-cut sparsifier having at most $s$ edges.
By Theorem~\ref{thm:sketchLB}, answering all possible cut queries correctly 
(in the sense of approximation factor $1\pm\eps$ with probability $1$) 
requires sketch size $\Omega(n/\eps^2)$ bits.
In fact, by inspecting the proof (specifically, of Lemma~\ref{lem:BobRecovers})
the above holds even if the correctness 
(1) holds only for cut queries contained in a single $G_j$,
i.e., queries $S\cup T$ for $S\subset L(G_j)$ and $T\subset R(G_j)$;
and (2) allows for each cut value 
an additive error of $\gamma/\veps^3$, where $\gamma=10^{-4}$.
The idea now is to encode $H$ using $m\approx s$ bits in a way 
that suffices to correctly answer all such cut queries
i.e., in the context of Lemma~\ref{lem:BobRecovers},
Alice will encode $H$ and send it as her sketch to Bob.
The sketch-size bound $m\geq \Omega(n/\eps^2)$ we proved for $G$ 
would then imply a similar bound on $s$.

Consider then the sparsifier $H$, which is an edge-weighted graph,
while the edges of $G$ all have unit weight.
Observe that $H$ must be a union of disjoint graphs $H_j$ on 
$L(G_j)\cup R(G_j)$ for $j\in[\veps^2n/2]$, 
because $H$ must preserve the corresponding cut, which has value zero.
Let $s_j$ denote the number of edges in $H_j$.
It will be convenient to consider each such graph separately,
so fix for now some $j\in[\veps^2n/2]$.

Consider first the case $s_j\leq \gamma^2/(6\eps^4)$,
and let us show how to encode $H_j$.
Construct from $H_j$ another graph $H'_j$ by rounding 
every non-integral edge weight
to one of its two nearby integers independently at random,
in an unbiased manner.
Specifically, each edge weight $w>0$ is rounded upwards to $w'=\ceil{w}$ 
with probability $w-\floor{w}$, and downwards to $w'=\floor{w}$ otherwise.
Now consider a fixed cut query $S\cup T$ in $G_j$,
denoting by $\delta_H(S\cup T)$ its cut value in $H$, and similarly for $H'$.
Then $\EX[\delta_{H'}(S\cup T)] = \delta_H(S\cup T)$,
and since the number of edges participating in this cut (in $H_j$) 
is at most $s_j\leq \gamma^2/(6\eps^4)$, 
by Hoeffding's inequality for $t=\gamma/(2\veps^3)$,
\[
  \Pr\Big[ \abs{\delta_{H'}(S\cup T) - \delta_H(S\cup T)} > t \Big]
  \leq e^{-2t^2/s_j} 
  \leq e^{-3/\eps^2}.
\]
Applying a union bound over at most $2^{2/\eps^2}$ possible cut queries $S\cup T$, 
we see that there exists $H'_j$ (it is in fact obtained with high probability)
such that for every cut query, the cut value in $H'_j$ is within additive $\gamma/(2\veps^3)$ from the cut value in $H_j$,
which in turn is within additive $\gamma\eps/2\cdot 1/\eps^4 = \gamma/(2\eps^3)$
from the cut value in $G_j$.
Hence, such $H'_j$ approximates all the cut values in $G_j$
sufficiently well for our intended application,
and Alice's sketch will thus encode $H'_j$ instead of $H_j$.
To simplify the description, let us include in $H'_j$ also edges of weight zero,
and then $H'_j$ has exactly $s_j$ edges (same as in $H_j$).
We further claim that the total edge-weight in $H'_j$ is at most $2/\eps^4$.
To see this, observe that 
(i) the total edge-weight in $H'_j$ (and similarly for $G_j$) 
is exactly twice the expected cut value of a random query in that graph;
and (ii) this expected cut value in $H'_j$ differs from the respective
expectation in $G_j$ by at most additive $\gamma/(2\veps^3)$.
It follows that the total edge-weight in $H'_j$ 
is at most $1/\eps^4$ larger than that in $G_j$,
which in turn is at most $1/\eps^4$.

The encoding of $H'_j$ has two parts, 
which describe separately the edges of $H'_j$ (without their weights), 
and their weights (assuming the edges are known).
Since $H'_j$ has $1/\eps^2$ vertices in each side, 
the number of possibilities for $s_j$ edges (without their weights)
among $\binom{2/\eps^2}{2} \leq 2/\eps^4$ vertex pairs 
is at most $\binom{2/\eps^4}{s_j}$.
Given the identity of $s_j$ edges in $H'_j$, 
the number of possibilities for their weights 
(recall the weights are integral and add up to at most $2/\eps^4$) 
is at most $\binom{s_j+2/\eps^4}{s_j} \leq \binom{4/\eps^4}{s_j}$.
We conclude that $H'_j$ can be encoded, on its two parts, 
using $O(\log \binom{4/\eps^4}{s_j})$ bits.

The second case $s_j > \gamma^2/(6\eps^4)$ is very simple,
and just encodes the original $G_j$ (instead of encoding $H_j$),
which trivially provides the value of every cut query inside $G_j$ exactly.
A straightforward encoding of $G_j$ takes $1/\eps^4$ bits.

Concatenating these encodings over all $j\in[\veps^2n/2]$,
yields a sketch that can approximate the value of all the needed cut queries 
(those that are inside a single $G_j$) within additive error $\gamma/\veps^3$.
It remains to bound the size of this sketch.
The number of $j$'s that fall into the second case $s_j > \gamma^2/(6\eps^4)$ 
is at most 
$\frac{\sum_j s_j}{\gamma^2/(6\eps^4)} = \frac{6\eps^4 s}{\gamma^2}$, 
and thus the total size of their encodings is at most $6s/\gamma^2$ bits.
For $j$'s that fall into the first case,
we use the fact $\binom{p}{k}\cdot \binom{p'}{k'} \leq \binom{p+p'}{k+k'}$,
and get that the total size of their encodings is at most
$ \sum_j O(\log \binom{4/\eps^4}{s_j}) 
  \leq  O(\log \binom{2n/\eps^2}{s}) 
  = O(s \log (\eps^{-2}n/s))
$ 
bits.
Altogether, there is a sketch of size $m=O(s(\gamma^{-2}+\log (\eps^{-2}n/s)))$ 
bits that encodes all the relevant cuts in $G$ within the desired accuracy.
Recalling that $\gamma$ is a constant and 
our sketch-size lower bound $m\ge \Omega(n/\eps^2)$
(from Lemma~\ref{lem:BobRecovers} and Theorem~\ref{thm:GHD}),
we conclude that the number of edges in $H$ is $s\geq \Omega(n/\eps^2)$,
which proves Theorem~\ref{thm:sparsifierLB}.

\section{Sketching Algorithm for Graph Cuts}
\label{sec:UB}

Let $G=(V,E,w)$ be a graph with $n=|V|$ vertices and edge weights $w:E\to\R_+$. 
The goal is to construct a sketch for $G$ that preserves each cut with constant
probability. For this goal we prove the following theorem.

\begin{theorem}
\label{thm:upper}
Given a weighted graph $G=(V,E,w)$ on $n$ vertices, where the non-zero
weights are in the range $[1,W]$, and $1/n\le \eps\le 1/30$, there
exists a cut sketch of size $\tilde O(n/\eps\cdot \log\log W)$
bits. Specifically, for every query $S\subset V$, the sketch produces
a $1+O(\eps)$ approximation to $w(S,\bar S)$, with probability at
least $7/9$.
\end{theorem}

In this section, we focus on the case when the non-zero edge weights
are in the polynomial range, i.e., $W=n^5$, as this is the crux of the
construction. We show how to extend the construction to the
general--weights case in Section~\ref{sec:largeWeights}.

We will say that $\hat a>0$ is a $\rho$-approximation to $a>0$ (for
$\rho\ge1$) if their ratio is $\hat a / a \in [1/\rho,\rho]$.  For an
edge weights function $w$, let $w(A,B)$ denote the total weight of
edges connecting two disjoint subsets $A,B\subset V$.

\subsection{Construction}

Below, we show how to sketch the graph, as well as how to estimate the
cut size given the sketch.

\paragraph{Sketching Algorithm.}
The sketch has two components. 
The first component is a standard $1.4$-cut sparsifier 
(recall that a $(1+\eps)$-cut sparsifier is a sparse graph $H$ on same vertex-set,
which approximates every cut in $G$ within a factor of $1+\eps$).
We can use the construction of Bencz{\'u}r and Karger \cite{BK96}, 
or subsequent constructions \cite{SS11,BSS12,FHHP11,KP12}
(some of which produce a spectral sparsifier, which is only stronger);
any of these methods will produce a graph $H$ with $\tO(n)$ edges. 

The second component is the main ingredient of the sketch, and is
described next.
Let $\tilde C=\{1.4^i\mid 0\le i\le \log_{1.4}n^5\}$ be the set of size
$O(\log n)$ such that each cut value in $G$ is $1.4$-approximated by
some value $c\in \tilde C$.  For each value $c\in \tilde C$ we
construct a structure $D_{c}$ as follows.
First by scaling all the edge weights, let us assume $c=1$. 
Now discard every edge $e$ whose (scaled) weight is $w_e>5$. 

In the next step, called importance sampling, we sample each
(remaining) edge $e\in E$ independently with probability $p_e\eqdef
\minn{w_e/\eps^2,1}$, and assign the sampled edges new edge weights
$\tilde w_e\eqdef w_e/p_e$.  Notice that $\tilde w_e\in [\eps^2,5]$.
(It may be convenient to consider the non-sampled edges as having
weight $\tilde w_e\eqdef 0$.)  Now let $\tilde E$ be the set of
sampled edges, and partition it into $l=O(\log \frac{1}{\eps})$
classes according to the (new) edge weights, namely, $\tilde
E=L_1\cup\cdots\cup L_l$ where $L_i = \aset{e\in E:\ \tilde w_e\in
  (5\cdot 2^{-i}, 5\cdot 2^{-i+1}] }$.

For each class $L_i$, recursively partition the graph $(V,L_i)$ as
follows: break the current cluster $P\subseteq V$ whenever it contains
a subset $P'\subset P$ of size $\card{P'}\leq \card{P}/2$ such that
$d_i(P',P\setminus P')/\card{P'}\leq 1/\eps$, where $d_i(A,B)$ is the
number (not weight) of class $L_i$ edges connecting between two
disjoint subsets $A,B\subset V$.
Once the recursive partitioning process is finished, denote the
resulting partition of $V$ by $\calP_i$, and store in the sketch all
the edges $E_{\calP_i}$ connecting different parts of $\calP_i$.  In
addition, store for every vertex $v\in V$ its ``weighted degree'',
i.e., the total weight of its incident class $L_i$ edges $\tilde
w_i(v)\eqdef \sum_{e\in L_i: v\in e} \tilde w_{e}$.  From the
remaining edges (those inside a part of $\calP_i$), we store the
following sample: for every non-isolated vertex store $s=1/\eps$
incident edges, each chosen uniformly at random with replacement.  For
every stored edge $e\in L_i$ (from $E_{\calP_i}$ or not), the sketch
keeps the edge weight $\tilde w_e$.  

\paragraph{Estimation algorithm.}
Given a query subset $S\subset V$, 
first use the graph sparsifier $H$ to compute $\tilde c$, 
a $1.4$-approximation to the desired cut value $w(S,\bar S)$,
and use the structure $D_{c}$, 
where $c\in \tilde C$ is a $1.4$-approximation to $\tilde c/(1.4)^2$,
namely, $c\in[\tilde c/(1.4)^3,\tilde c/1.4]$. 
Now estimate the contribution to the cut from edges in each class $L_i$
as follows.

The contribution from a single class $L_i$, $i\in [l]$ is composed of
two terms.  The first term is the total $\tilde w_e$ weight of edges
between $S$ and $\bar S$ that are in the set $E_{\calP_i}$ (recall the
sketch stores all edges between different parts of $\calP_i$).  

The second term estimates the $\tilde w_e$ weight of edges inside each
part $P\in \calP_i$.  Specifically, check whether $|S\cap P|\le |P|/2$
or $|\bar S\cap P|< |P|/2$.  In the first case, the estimate $I_P$ for
$P$ is given by the sum over all vertices in $S\cap P$, of their
``weighted degree'' (which is stored in the sketch) minus the
(appropriately scaled) weight $\tilde w$ of sampled edges that are
inside $S\cap P$.
Formally, we define
\begin{equation}  \label{eq:I_p}
  I_P \eqdef \sum_{x\in S\cap P} \Big[\tilde w_i(x) - \tfrac{d_{ix}(P)}{s}\sum_{e\in L_{ix}} 1_{\aset{e\in E(S,\bar S)}}\tilde w_e \Big],
\end{equation}
where $d_{ix}(T)$ is the number of class $L_i$ edges from vertex $x$ to a set $T\subset V$, 
and $L_{ix}$ is the multiset of $s$ class $L_i$ edges incident to $x$
(the sample chosen by the sketching algorithm).
In
the second case ($|\bar S\cap P|<|P|/2$), the estimate is similar
except that we now use the ``weighted degrees'' in $\bar S\cap P$ and
total weight of edges inside $\bar S\cap P$.  

The overall second term is the sum of these estimates over all parts
$P\in \calP_i$.  The final estimate is just the sum of these two terms
over all classes $L_i$.

\subsection{Size and Correctness Guarantees}

First we show the bound on space usage.  The sparsifier $H$ has
$\tO(n)$ edges.  By construction, we have $O(\log n)$ possible cut
values $\tilde C$ and for each one, we have $l=O(\log
\tfrac{1}{\eps})\leq O(\log n)$ edge weight classes.  For each weight
class $L_i$, the sketch stores (1) at most $O(\tfrac{1}{\eps}n\log n)$
edges in $E_{\calP_i}$, because each step in the recursive
partitioning process contributes $d_i(P',P\setminus P')/\card{P'}\leq
1/\eps$ edges per vertex in $P'$, and each vertex appears in the
smaller subset $P'$ at most $\log n$ times; and (2) at most $n/\eps$
sampled edges, because for every non-isolated vertex we sample
$s=1/\veps$ incident edges.  All in all, we have $\tO(n/\eps)$ edges,
each requiring $O(\log n)$ bits.

We proceed to analyze the accuracy of the estimation procedure.  Fix a
query $S\subset V$.  After consulting the sparsifier $H$, we can
approximate $w(S,\bar S)$ within factor $(1.4)^2<2$ and thus use a
data structure $D_c$ where $c\le w(S,\bar S)\le 4c$.  Thus, by
rescaling to $c=1$, we need to $1+\eps$ estimate the cut value
$w(S,\bar S)$ which is between $1$ and $4$.

First, note that the discarded edges do not affect the solution since
they could not have been part of $w(S,\bar S)$ as they are too heavy.
Second, we bound the effect of the importance sampling step on the value 
of the cut, namely, that \whp $\tilde w(S,\bar S)$ is a $1+\eps$
approximation to $w(S,\bar S)$. 
Indeed, its expectation $\EX[\tilde w(S,\bar S)] = w(S,\bar S)$ 
since every edge $e$ that was not discarded satisfies $\EX[\tilde w_e] = w_e$
(and more generally, it is a Horvitz-Thompson estimator),
and its variance is
\[
  \Var[\tilde w(S,\bar S)]
  = \sum_{e\in E(S,\bar S)} \Var[\tilde w_e]
  = \sum_{e\in E(S,\bar S)} w_e^2/p_e-w_e^2
  \leq \sum_{e\in E(S,\bar S)} \eps^2 w_e
  = \eps^2 w(S, \bar S),
\] 
where the inequality is verified separately for $p_e=1$ and for $p_e=w_e/\eps^2$.
Thus, by Markov's inequality, with probability at least $8/9$, we have
$\abs{\tilde w(S,\bar S) - w(S,\bar S)} 
  \leq 3 \veps \sqrt{w(S, \bar S)}
  \leq 3 \veps \cdot w(S, \bar S)$.

Next, we show that the rest of the procedure is likely to estimate 
$\tilde w(S,\bar S)$ well. 
Let $c_i$ be the contribution to $\tilde w(S,\bar S)$ by edges in class $L_i$; 
hence $\sum_{i=1}^l c_i=\tilde w(S,\bar S)$ 
and each $c_i\le \tilde w(S,\bar S)\leq (1+3\veps)\, w(S,\bar S)\le 5$. 
Let us also denote $\lambda_i\eqdef 5\cdot 2^{-i}$, 
then all edges $e\in L_i$ have weight $\tilde w_e\in [\lambda_i, 2\lambda_i]$. 
Let $\hat c_i$ be the contribution of edges of class $L_i$ to the estimate
computed by the algorithm;
it is equal to $\sum_{e\in E_{\calP_i}\cap E(S,\bar S)} \tilde w_e$,
which is computed exactly because the sketch stores all edges between 
different parts of $\calP_i$,
plus the sum over all parts $P\in \calP_i$ of their estimate $I_P$.
Recall that, in \eqref{eq:I_p}, each $I_P$ is the sum of ``weighted degrees'' 
in either $S\cap P$ or $\bar S\cap P$ (whichever has smaller cardinality),
minus the appropriately scaled weight of edges inside that subset.

It is easy to verify that $\hat c_i$ is an unbiased estimator for $c_i$, 
namely, $\EX[\hat c_i]=c_i$. 
We now analyze its variance, which comes only from 
the estimators $I_P$ of edges inside each $P\in\calP_i$.
We can assume that for each $P\in \calP$ we have $\card{S\cap P}\le \card{P}/2$ 
(otherwise, exchange $S\cap P$ and $\bar S\cap P$). 
Let $c_i(P)$ be the weight of the cut $\tilde w(S,\bar S)$ inside $P$, 
restricted to edges from class $L_i$. Then $\sum_{P\in\calP_i} c_i(P) \leq c_i$.
In each part $P$, the variance of $I_P$ comes only from the sampled edges 
(since the ``weighted degrees'' are known exactly):
\begin{align*}
  \Var[I_P]
  &= \sum_{x\in S\cap P} \left( \tfrac{d_{ix}(P)}{s} \right)^2 \Var\Big[ \sum_{e\in L_{ix}} 1_{\aset{e\in E(S,\bar S)}} \tilde w_e \Big] \\
  &\leq \sum_{x\in S\cap P} \left( \tfrac{d_{ix}(P)}{s} \right)^2\cdot s\cdot \tfrac{d_{ix}(S\cap P)}{d_{ix}(P)}\cdot (2\lambda_i)^2 \\
  &= 4\veps\lambda_i^2\sum_{x\in S\cap P} d_{ix}(P)\cdot d_{ix}(S\cap P) \\
  &\le 4\veps\lambda_i^2\cdot |S\cap P| \sum_{x\in S\cap P} d_{ix}(P) \\
  &\le 4\veps\lambda_i^2\cdot |S\cap P|\cdot 
      \Big[ d_i(S\cap P, \bar S\cap P) + 2 |S\cap P|^2 \Big].
\intertext{By the stopping condition of the recursive partitioning
$\card{S\cap P} \leq \veps\cdot d_i(S\cap P,\bar S\cap P)$, and thus
}
  \Var[I_P]
  &\le 4\veps^2\lambda_i^2 (d_i(S\cap P,\bar S\cap P))^2 
       \Big[ 1 + 2 \veps^2 d_i(S\cap P,\bar S\cap P) \Big].
\intertext{Observe that $c_i(P)$ consists of $d_i(S\cap P,\bar S\cap P)$ edges 
whose $\tilde w$ weight is in the range $[\lambda_i,2\lambda_i]$,
hence $\lambda_i d_i(S\cap P,\bar S\cap P) \leq c_i(P)$.
Using $c_i(P)\leq c_i\leq 5$ and $\lambda_i \geq \veps^2$,
we can further derive $\veps^2 d_i(S\cap P,\bar S\cap P) \leq 5$,
which together give 
}
  \Var[I_P]
  &\le 44\veps^2 (c_i(P))^2.
\end{align*}
Therefore, the total variance over all parts $P\in\calP_i$ is 
\[
  \Var[\hat c_i]
  = \sum_{P\in\calP_i} \Var[I_P]
  \leq 44\veps^2 \sum_{P\in\calP_i} (c_i(P))^2
  \leq 44\veps^2 c_i^2.
\]
It follows that the algorithm's final estimate
$\hat c\eqdef \sum_{i=1}^l \hat c_i$ is unbiased, namely,
\[
  \EX[\hat c] 
  = \sum_i \EX[\hat c_i]
  = \sum_i c_i
  = \tilde w(S,\bar S),
\]
and its total variance (over all levels $i$) is 
\[
  \Var[\hat c]
  = \sum_i \Var[\hat c_i]
  \leq 44\veps^2 \sum_i c_i^2
  \leq 44\veps^2 \tilde w(S,\bar S).
\]
Thus, by Markov's inequality, with probability at least $8/9$, we have that
$
  \abs{\hat c - \tilde w(S,\bar S))} 
  \leq 3 \veps \sqrt{44\, \tilde w(S,\bar S)}
  \leq 21 \veps\, \tilde w(S,\bar S)
$,
where the last inequality uses that
$\tilde w(S,\bar S) \geq (1-3\veps)\, w(S,\bar S) \geq 9/10$.
Altogether, we obtain that with probability at least $7/9$, 
the algorithm's estimate $\hat c$ 
is a $1+O(\eps)$ approximation for $\tilde w(S,\bar S)$,
and also the latter quantity is a $1+O(\eps)$ approximation for $w(S,\bar S)$,
Or, more directly, 
\[
  \abs{\hat c-w(S,\bar S)}
  \leq 3\veps\, w(S,\bar S) + 21 \veps (1+3\veps)\, w(S,\bar S)
  \leq 27 \veps\, w(S,\bar S),
\]
which completes the proof of Theorem ~\ref{thm:upper}.

\section{Extending the Sketching Algorithm to General Edge Weights}
\label{sec:largeWeights}

We now build on the results of Section~\ref{sec:UB} for polynomial weights
to show the upper bound (Theorem \ref{thm:upper}) for general edge weights.
Concretely, assume there is a sketching algorithm, 
which we shall call the ``basic sketch'',
for the case where all edge weights are in a polynomial range, 
say for concreteness $[1,n^5]$ 
(which by scaling is equivalent to the range $[b,n^5b]$ for any $b>0$),
which uses space $\tO(n/\veps)$.
We may assume the success probability of this sketch is at least $1-1/n^8$,
e.g., by using standard amplification using $O(\log n)$ repetitions,
thereby increasing the sketch size by at most $O(\log n)$ factor.
Throughout, we measure the memory usage in machine words, and assume
that a machine word can accommodate the weight of an edge and also at
least $2\log n$ bits.  As before, we may assume $\veps>1/n$, as
otherwise the theorem is trivial.

\subsection{Construction}

\paragraph{Sketching Algorithm.}
The sketch has two components. 
The first component is essentially a maximum-weight spanning tree $T$
computed using Kruskal's algorithm. 
Specifically, start with $T$ as an empty graph on vertex set $V$ 
(so every vertex forms a connected component of size one), 
and then go over the edges $e\in E$ in some order $\pi$ of decreasing weight,
each time adding $e$ to the current $T$ if this would not close a cycle in $T$
(i.e., the endpoints of $e$ are currently in different connected components).
If the edge weights are all distinct, then $\pi$ is unique;
otherwise, fix some $\pi$ by breaking ties arbitrarily
(and use the same $\pi$ later, for sake of consistency).
At the end, $T$ is a spanning tree and has $n-1$ edges.
Our sketch stores the list of edges that form $T$,
denoted $e_1,e_2,\ldots,e_{n-1}\in E$, sorted in order of insertion 
(which is also their ordering according to $\pi$).

The second component is computed by iterating over $e_1,\ldots,e_{n-1}$;
in iteration $j$, we consider the graph obtained from $G$ in three steps:
(i) remove all edges $e\in E$ of weight $w(e) < w(e_j)/n^3$; 
(ii) change all edges $e\in E$ of weight $w(e)\geq n^2\cdot w(e_j)$ to have
infinite weight; and
(iii) contract all edges of infinite weight, keeping parallel edges 
(self-loops may be removed as they will have no effect anyway).%
\footnote{We manage vertex names in a systematic manner, 
e.g., a merging of vertices $u,v$ keeps the lexicographically smaller name.
}
We denote the graph obtained after step (ii) by $G_j$,
and the one obtained after step (iii) by $G'_j$.
Notice that from the perspective of cuts, $G_j$ and $G'_j$ are equivalent,
and we thus refer to $G'_j$ as the \emph{reduced form} of $G_j$.
Observe that $G'_j$ has at most $n$ vertices, 
and its edge weights lie in the range $[n^{-3}\cdot w(e_j), n^2\cdot w(e_j)]$,
so in principle, we can apply the assumed sketching algorithm 
(that works for edge weight in a polynomial range) on $G'_j$.
We do so (apply the assumed sketching algorithm) 
and store its result in our sketch, but with two twists:
First, we apply the assumed sketching algorithm 
separately on every connected component of $G'_j$ of size at least $2$
rather than on the entire $G'_j$.
Second, we do it only if there is no earlier iteration $k<j$ 
with $w(e_k)/w(e_j)<2$ for which we already sketched and stored $G'_k$.
For instance, if $w(e_{j-1})=w(e_j)$, then $G'_{j-1}$ and $G'_j$ are identical
and $G'_j$ will not be stored (because of $k=j-1$ or an even smaller $k$).

\paragraph{Estimation algorithm.}
Given a query subset $S\subset V$, 
find the smallest $j\in[n-1]$ such that $e_j$ crosses the cut $(S,\bar S)$; 
such $j$ exists because $\aset{e_1,\ldots,e_{n-1}}$ forms a spanning tree.
We further show in Lemma \ref{lem:heaviest} below that $e_j$ 
is a heaviest edge in this cut, hence $w(S,\bar S) / w(e_j) \in [1,n^2]$.
Now find the largest $k\leq j$ for which we sketched and stored $G'_k$;
by construction $w(e_k)/w(e_j)\in[1,2)$.
Lemma \ref{lem:GkApprox} below proves that the cut values in $G$ and in $G_k$ 
are almost the same.
Next, compute the connected components of the graph $(V,\aset{e_1,\ldots,e_k})$,
and observe they must be exactly the same as the connected components of $G_k$.
Obviously, the value of the cut $(S,\bar S)$ in $G_k$ 
is just the sum, over all connected components $V'\subset V$ in $G_k$, 
of the contribution to the cut from edges inside that component, 
namely $w(S\cap V', \bar S\cap V')$.
Recall that $G'_k$ has essentially the same cuts as $G_k$
and we can thus estimate each such term $w(S\cap V', \bar S\cap V')$ 
using the sketch we prepared for $G'_k$
(more precisely, using the sketch of the respective component $V'$ of $G'_k$, 
unless $\card{V'}=1$ in which case that term is trivially $0$).
To this end, we need to find out which vertices of $G_k$ 
were merged together to form $G'_k$, 
which can be done using $e_1,\ldots,e_{n-1}$ as follows.
Find the largest $k^*$ such that $w(e_{k^*}) \geq n^2\cdot w(e_k)$,
and compute the connected components of the graph $(V,\aset{e_1,\ldots,e_{k^*}})$.
Lemma \ref{lem:ConnComp} below proves that these connected components 
(or more precisely the partition of $V$ they induce) 
are exactly the subsets of vertices that are merged in $G_k$ to create $G'_k$.
Now that know the vertex correspondence between $G_k$ and $G'_k$,
we estimate the cut value $w(S\cap V', \bar S\cap V')$ 
by simply using the estimate for the corresponding cut value in $G'_k$, 
where the latter is obtained using the basic sketch prepared for $G'_k$.

\subsection{Accuracy Guarantee}

The accuracy of the estimation algorithm follows from the above discussion,
which uses the three lemmas below, 
together with a union bound over the events of an error in any of the 
basic estimates used along the way, the number of which is $O(n)$, 
because they correspond to disjoint subsets of $V$.
(The union bound is applicable because these basic sketch are queried 
in a non-adaptive manner, or alternatively, because we make at most
one query to every basic sketch that is constructed independently of the others.)

\begin{lemma} \label{lem:heaviest}
Fix $S\subset V$ and let $e'\in E$ be the first edge, 
according to the ordering $\pi$, that crosses the cut $(S,\bar S)$. 
Then this $e'$ is the first edge in the sequence $e_1,\ldots,e_{n-1}$
that crosses the cut $(S,\bar S)$.
\end{lemma}
\begin{proof}
Let $e'\in E$ be the first edge, according to the ordering $\pi$,
that crosses the cut $(S,\bar S)$. 
Clearly, $e'$ is a heaviest edge in this cut.
Now observe that in the construction of $T$ (i.e., $e_1,\ldots,e_{n-1}$),
when $e'$ is considered, $T$ has no edges between $S$ and $\bar S$,
hence the endpoints of $e$ lie in different connected components,
and $e'$ must be added to $T$.
\end{proof}
\begin{lemma} \label{lem:GkApprox}
Consider a query $S\subset V$ and let $k\in[n-1]$ be the value computed
in the estimation algorithm. 
Then the ratio between the value of $w(S,\bar S)$ in the graph $G_k$ 
and that in the graph $G$ is in the range $[1-\tfrac1n,1]\subset [1-\veps,1]$,
formally
\[
   1-\tfrac1n \leq \frac{w_{G_k}(S,\bar S)}{w_{G}(S,\bar S)} \leq 1.
\]
\end{lemma}
\begin{proof}
The edges in $G_k$ are obtained from the edges of $G$, by either 
(1) removing edges $e$ whose weight is $w(e)<w(e_k)/n^3$; or 
(2) changing edges $e$ with $w(e)\ge n^2\cdot w(e_k)$ to have infinite weight.
The first case can only decrease any cut value, %
while the second case can only increase any cut value.

Recall that the estimation algorithm finds $j$ such that 
$w(S,\bar S)/w(e_j) \in [1,n^2]$, and then finds $k\leq j$, 
which we said always satisfies $w(e_k)/w(e_j) \in [1,2)$. 
Thus, $w(S,\bar S)/w(e_k) \in (\half,n^2]$.
So one direction of the desired inequality follows by observing 
that edges in $G$ that fall into case (1) have the total weight at most
\[
  \tbinom{n}{2} w(e_k)/n^3 
  \leq \tfrac{2}{n} w(e_k) 
  \leq \tfrac{1}{n} w(S,\bar S).
\]
The other direction follows by observing that edges $e$ that fall into case (2)
have (in $G$) weight $w(e) > n^2\cdot w(e_k) \ge w(S,\bar S)$,
and therefore do not belong to the cut $(S,\bar S)$.
\end{proof}

\begin{lemma} \label{lem:ConnComp}
Fix $w^*>0$, let $E^*\eqdef\aset{e\in E: w(E) \geq w^*}$,
and find the largest $i^*\in[n-1]$ such that $w(e_{i^*}) > w^*$.
Then the graphs $(V,E^*)$ and $(V,\aset{e_1,\ldots,e_{i^*}})$ 
have exactly the same connected components 
(in terms of the partition they induce of $V$).
\end{lemma}
\begin{proof}
It is easy to see that executing our construction of $T$ above on the set $E^*$,
gives the exact same result as executing it for $E$ but stopping once we reach 
edges of weight smaller than $w^*$.
The latter results with the edges $e_1,\ldots,e_{i^*}$,
while the former is clearly an execution of Kruskal's algorithm,
i.e. computes a maximum weight forest in $E^*$.
\end{proof}

\subsection{Size Analysis}

\begin{lemma} \label{lem:size}
The total size of the sketch is at most $\tO(n/\veps\log\log W)$,
where we assume all non-zero edge weights are in the range $[1,W]$.
\end{lemma}
\begin{proof}
The first component of the sketch is just a list of $n-1$ edges with
their edge weights, hence its size is $O(n\log (\log W/\eps))$ (we can
store a $1+\eps/2$ approximation to each weight using space $\log
(\log W/\eps)$).

The second component of the sketch has $n-1$ parts, 
one for each $G'_j$ where $j\in[n-1]$.
For some of these $j$ values, we compute and store the basic sketch 
for every connected components of $G'_j$ that is of size at least $2$.
Denoting by $n_j$ the number of vertices in $G'_j$,
and by $m_j$ the number of connected components in $G'_j$,
the storage requirement for each $G'_j$ 
is at most $\tO(\tfrac{n_j-m_j}{\veps})$,
because each connected component of size $s\ge 2$ 
requires storage $\tO(\tfrac{s}{\veps}) \leq \tO(\tfrac{s-1}{\veps})$,
and these sizes (the different $s$ values) add up to at most $n$.

Denote the values of $j$ for which we do store a basic sketch for $G'_j$
by $j_1<j_2<\cdots<j_p$, 
where by construction $w(e_{j_i}) / w(e_{j_{i+1}}) \geq 2$.
Summing over these values of $j$, the second component's storage requirement 
is at most 
\begin{equation} \label{eq:storage1}
\tO(\tfrac{1}{\veps} \sum_{i\in[p]}(n_{j_i}-m_{j_i}) ).
\end{equation}
To ease notation, let $M\eqdef 5\log_2 n$,
and consider the graphs $G'_{j_i}$ and $G'_{j_{i+M}}$ for some $[i\in p-M-1]$.
Observe that every edge in $G'_{j_i}$ has weight at least 
$w(e_{j_i})/n^3 \geq 2^M \cdot w(e_{j_{i+M}})/n^3 = n^2 \cdot w(e_{j_{i+M}})$
(because edges of smaller weight are removed);
thus, in $G_{j_{i+M}}$, these same edges have infinite weight, 
and then to create the reduced form $G_{j_{i+M}}$, these edges are contracted.
It follows from this observation that every connected component in $G'_{j_i}$
becomes in $G_{j_{i+M}}$ a single vertex, hence $n_{j_{i+M}} \leq m_{j_i}$
(we do not obtain equality since additional contractions may occur).
Using this last inequality, for every $i^*\in[M]$,
we can bound the following by a telescopic sum
\[
  \sum_{i=i^*,i^*+M,i^*+2M,\ldots} (n_{j_i}-m_{j_i})
  \leq \sum_{i=i^*,i^*+M,i^*+2M,\ldots} (n_{j_i}-n_{j_{i+M}})
  \leq n_{j_{i^*}}
  \leq n,
\]
and therefore
\[
  \sum_{i\in[p]}(n_{j_i}-m_{j_i}) )
  \leq \sum_{i^*\in[M]} \sum_{i=i^*,i^*+M,i^*+2M,\ldots} (n_{j_i}-m_{j_i})
  \leq M\cdot n.
\]
Plugging this last inequality into \eqref{eq:storage1}, 
we obtain that the second component's storage requirement is at most
$M\cdot \tO(n/\veps)$, which is still bounded by $\tO(n/\veps)$.
\end{proof}

\section{``For Each'' Sketch Requires Size $\Omega(n/\eps)$}
\label{sec:nEpsLB}

The following theorem proves that our sketch from Theorem \ref{thm:upper} 
achieves optimal space up to polylogarithmic factors, 
even for unweighted graphs.

\begin{theorem}
Fix an integer $n$ and $\eps\in[2/n,1/2]$. Suppose $\sk$ 
is a sketching algorithm that outputs at most $s=s(n,\eps)$ bits,
and $\est$ is an estimation algorithm, such that together,
for every $n$-vertex graph $G$ and subset $S\subset V$, 
with probability at least $9/10$ 
the estimation procedure is correct up to factor $1+\eps$, i.e., 
\[
  \Pr\Big[\est(S,\sk(G))\in (1\pm \eps)\cdot\card{E(S,\bar S)}\Big]\ge 9/10.
\]
Then $s \ge \Omega(n/\eps)$.
\end{theorem}
\begin{proof}
We will show how to encode a bit-string of length $l\eqdef n/(8\eps)$ into a
graph, so that, given its sketch $\sk(G)$, one can reconstruct any bit
of the string with constant probability. 
Standard information-theoretical argument would
then imply that $s\geq\Omega(l) = \Omega(n/\eps)$.

Given a string $x\in \{0,1\}^l$, we embed it into a bipartite graph
$G$ on with $n/2$ vertices on each side, 
and vertex degrees bounded by $D\eqdef 1/(4\eps)$ as follows. 
Partition the vertices on each side into disjoint blocks of $D$,
and let the $i$-th block on the left side and on the right side 
form a (bipartite) graph which we call $G_i$, for $i=1,\ldots,n/(2D)$. 
Then partition the string $x$ in $n/(2D)$ blocks,
each block is of length $D^2$ and describes the adjacency matrix 
of some bipartite $G_i$.

We now show that evaluating a bit from the string $x$ corresponds to testing 
the existence of some edge $(u,v)$ from some graph $G_i$, which we can do 
using the $1+\eps$ approximating sketch only. 
Formally, let $\delta(S)$ be the cut
value of the set $S$, i.e., $\card{E(S,\bar S)}$,
and observe that 
\[
  \delta(\{u\})+\delta(\{v\})-\delta(\{u,v\}) =
  \begin{cases}
    0 & \text{if $(u,v)$ is an edge in $G$;} \\
    2 & \text{otherwise}.
  \end{cases}
  \]
Since the considered values of $\delta(\cdot)$ are bounded by $D$, 
the sketch estimates each such value with additive error at most $\eps D=1/4$, 
which is enough to distinguish between the two cases.
Furthermore, since we query the sketch only $3$ times, the
probability of correct reconstruction of the bit is at least $7/10$. 
The lower bound follows.
\end{proof}

{\small 
\bibliographystyle{alphaurlinit}
\bibliography{robi} 
}

\appendix

\end{document}